\documentclass[a4paper, 11pt]{article}

\usepackage{amsmath,amsfonts,amssymb,amsthm,graphicx,graphics,epsf}
\usepackage{graphicx}

\usepackage{algorithmic}
\usepackage{algorithm}
\usepackage{hyperref,color}
\usepackage{xspace}

\usepackage{geometry}

\newtheorem{theorem}{Theorem}
\newtheorem{lemma}[theorem]{Lemma}

\theoremstyle{definition}

\theoremstyle{remark}
\newtheorem*{remark}{Remark}

  \newcommand{\R}{\mathbb{R}}

\begin{document}

\title{On the Complexity of Polytopes in $LI(2)$}

\author{Komei Fukuda \thanks{Department of Mathematics and Department of Computer Science. Institute of Theoretical Computer Science,
 ETH Z\"{u}rich.  CH-8092 Z\"{u}rich, Switzerland.
  \texttt{komei.fukuda@math.ethz.ch}}
  \and 
  May Szedl{\'a}k\thanks{Department of Computer Science. Institute of Theoretical Computer Science, ETH Z\"{u}rich.
  CH-8092 Z\"{u}rich, Switzerland.
    \texttt{may.szedlak@inf.ethz.ch} \newline Research supported by the Swiss National Science Foundation (SNF Project 200021\_150055 / 1)}
  }

\maketitle
\begin{abstract}
In this paper we consider polytopes given by systems of $n$ inequalities in $d$ variables, where every inequality has at most two variables with nonzero coefficient. We denote this family by $LI(2)$. We show that despite of the easy algebraic structure, polytopes in $LI(2)$ can have high complexity. We construct a polytope in $LI(2)$, whose number of vertices is almost the number of vertices of the dual cyclic polytope, the difference is a multiplicative factor of depending on $d$ and in particular independent of $n$. Moreover we show that the dual cyclic polytope can not be realized in $LI(2)$.
\end{abstract}

\section{Introduction}

Throughout, we assume that we are given a bounded polytope by a system of $n$ inequalities in $d$ variables, of form $P = \{x \in \R^d \mid Ax \leq b\}$, where $A \in \R^{n \times d}$ and $b \in \R^n$. 
In the feasibility problem we want to find a solution $x \in P$.
In general no strongly polynomial time algorithm (polynomial in $d$ and $n$) to solve the feasibility problem is known. Although the simplex algorithm runs fast in practice, in general it can have exponential running time \cite{dantzig-simplex, klee-minty-simplex}. On the other hand the ellipsoid method runs in polynomial time on the encoding of the input size, but is not practical \cite{Khachiyan}. A first practical polynomial time algorithm, the interior-point method, was introduced in \cite{Karmarkar}, and has been modified in many ways since \cite{wright97}.

We denote by $LI(2)$ the family of systems $Ax \leq b$, that have at most two variables per inequality with nonzero coefficient. 
In this family, Hochbaum and Naor's algorithm finds a feasible point or a certificate for infeasibility in time $O(d^2n \log n)$ \cite{HN}, i.e., it solves the feasibility problem in strongly polynomial time. Using this result and Clarkson's redundancy removal algorithm \cite{c-mosga-94}, it was shown that in $LI(2)$ all redundancies can be detected in strongly polynomial time $O(n d^2 s \log s)$, where $s$ denotes the number of nonredundant constraints \cite{twovariables}. 
Because of this difference in running time, it is hence natural to ask, whether polytopes $LI(2)$ have a simpler structure than general polytopes. In particular we are interested to know how many vertices a polytope of this family can have. 

It is known that in general the dual cyclic polytope maximizes the number of vertices for a polytope given by $n$ constraints (see Theorem \ref{thm_MM}).
In this paper we construct a polytope in $LI(2)$, that has almost the same complexity as the dual cyclic polytope. This polytope was already introduced in \cite{Amenta1999} in the context of deformed products. 
In this polytope the number of vertices is smaller by a factor that only depends on the dimension $d$ and not on $n$, (see Lemma \ref{lemma_vert}). A similar result can be shown not only for vertices but for all $k$-faces (see Theorem \ref{thm_mainLI2}). 
This shows that polytopes in $LI(2)$ can have high complexity; if $d$ is constant, then even the same complexity as the dual cyclic polytope.   

We will also show in Theorem \ref{thm_slack} that the dual cyclic polytope can not be realized in LI(2) for $d \geq 4$. In particular in the dual cyclic polytope any pair of the $n$ facets are adjacent, however in $LI(2)$, there are $\Omega({n^2}/d^2)$ pairs that are not adjacent.

\section{Definitions and Known Results}
Let $P = \{x \in \R^d \mid Ax \leq b\}$ be a convex polytope in $\R^d$, where $A \in \R^{n \times d}$ and $b \in \R^n$. The rows of $Ax \leq b$ are called the \emph{constraints}. 
The \emph{dimension} of $P$, denoted $\dim(P)$, is defined as the number of affinely independent points in $P$ minus one. 
A $k$-dimensional subset $F \subseteq P$ is a $k$-face of $P$, if $F$ has dimension $k$ and if there exists a hyperplane $h: ax \leq b$, such that $ax^* = b$ for all $x^* \in F$ and $ax^* < b$ for all $x^* \in P \setminus F$. The $0$-dimensional faces are called the \emph{vertices} of $P$, the $(d-1)$-dimensional faces are called \emph{facets}. If $F$ is a $k$-face, then $ax = b$ for at least $d-k$ constraints of $Ax \leq b$.

For $0 \leq k \leq d$ we denote by $f_k := f_k(P)$ the number of $k$-dimensional faces of $P$. The \emph{$f$-vector} of $P$ is defined by
$f(P) := (f_0, f_1, \dots, f_d).$

\begin{theorem} [McMullen's Upper Bound Theorem \cite{McMullen, Fukuda2016}] 
\label{thm_MM}
The maximum number of $k$-faces in a $d$-dimensional polytope with $n$ constraints is attained by the dual cyclic polytope $c^*(n,d)$ and is given by
\begin{align*}
f_k(c^*(n,d)) = \sum_{r = \min \{k, \lceil d/2 \rceil\}}^{\lceil d/2 \rceil -1} \binom{n-d-1 + r}{r}\binom{r}{k} + \sum_{r = \max \{k, \lceil d/2 \rceil\}}^{d} \binom{n-r-1}{d- r}\binom{r}{k}.
\end{align*}
\noindent In particular the number of vertices is given by
\[f_0 (c^*(n,d)) = \binom{n- \lceil d/2 \rceil}{n-d} + \binom{n- \lfloor d/2 \rfloor -1}{n-d}.\]
\end{theorem}

\begin{remark}
For $\lceil d/2 \rceil \leq k \leq d$ the formula can be simplified to
\[f_k(c^*(n,d)) = \binom{n}{d-k}.\]
This means that any $(d-k)$ constraints define a $k$-face. 
\end{remark}

For our calculation we will make use of the following well known formulas.
Stirling's formula says that 
\[n!  = \Theta\left(\sqrt{n} \frac{n^n}{e^n}\right),\]
as $n$ goes to infinity.
It follows that 
\begin{equation}
\label{eq_binom}
\binom{n}{k} \leq O(1) \cdot \frac{n^n}{k^k(n-k)^{n-k}}.
\end{equation}
Furthermore we need the well known inequality
\begin{equation}
\label{eq_e}
1+x \leq e^x, \text{ for all } x \in R.
\end{equation}
We conclude that 
\begin{equation}
\label{eq_2}
\binom{n}{k} \leq O(1) \cdot \left(\frac{n}{k}\right)^k \cdot \left(1 + \frac{k}{n-k}\right)^{n-k} \leq O(1) \cdot \left(\frac{n}{k}\right)^k \cdot e^k.
\end{equation}

\section{Lower Bound on Maximum Complexity of $LI(2)$}
\sectionmark{Lower Bound on Maximum Complexity}
In the following we always assume that $\lfloor d/2 \rfloor$ is a divisor of $n$ (if $d$ is even) or $n-1$ (if $d$ is odd). All results naturally extend to any $d<n$, but we would like to avoid to have even more floors and ceilings in the notation. 
We want to construct a polytope in $LI(2)$, that has high complexity, i.e., with an $f$-vector of order close to the $f$-vector of the dual cyclic polytope. 

In a first part let us assume that $d$ is even. We pair the set of variables and define an $n/(d/2)$ polygon on each of the pairs. Formally, for $1 \leq i \leq d/2$, let 
\[A_i \begin{pmatrix}
 x_{2i-1}\\
 x_{2i}
\end{pmatrix} \leq b_i,\]
be a polygon in the $(x_{2i-1}, x_{2i})$-plane, given by $n/(d/2)$ constraints with $n/(d/2)$ vertices. 
We denote $P_i^*:=P_i^*(n,d) = \{x \in \R^2 \mid A_i (x_{2i-1}, x_{2i})^T \leq b_i\}$ and by $G_i$ the set of constraints of $P_i^*$.

Now $P^*(n,d)$ is defined as the $d$-dimensional polytope that we obtain from the union of $G_i$, $1 \leq i \leq d/2$. Since the $P_i^*$'s do not share any variables,
\[P^*(n,d) = \{x \in \R^d \mid (x_{2i-1}, x_{2i}) \in P_i, \text{ for all } 1 \leq i \leq d/2\}.\]

For $d$ odd, we pair the first $d-1$ variables and use the construction as above. Moreover we add the constraint $x_d \geq 0$, i.e.,
\[P^*(n,d) = \{x \in \R^d \mid (x_1, \dots, x_{d-1}) \in P^*(n-1, d-1) \wedge x_d \geq 0\}.\]

\begin{theorem} \cite{Amenta1999}
\label{thm_vert}
For $d$ even, the polytope $P^*(n,d)$ in $LI(2)$ has the following number of vertices:
\[\left(\frac{n}{ d/2 }\right)^{ d/2 }.\]
For $d$ odd it is
\[\left(\frac{n-1}{\lfloor d/2 \rfloor}\right)^{\lfloor d/2 \rfloor}.\]
\end{theorem}

The proof of \cite{Amenta1999} is given in a much more general setting of deformed products, we will here give the proof for our special case.
\begin{proof}
Let us assume first that $d$ is even. 
For $1 \leq i \leq d/2$ let $G_i =\{ g_i^1, \dots , g_i^{n/(d/2)}\}$, where the $g_i^j: a_i^j (x_{2i-1}, x_{2i})^T \leq b_i^j$, ordered in such a way that $g_i^j$ and $g_i^{(j+1)}$, $1 \leq j \leq d/2$, define a vertex of $P_i^*$. Throughout the proof, $j+1$ is always considered modulo $n/(d/2)$. We will show that if for every $P_i^*$ we choose two consecutive constraints $g_i^j$ and $g_i^{j+1}$, 
 these $d$ constraints define a vertex of $P^*(n,d)$ and those are the only sets of $d$ constraints that define vertices (see also Figure \ref{fig_ex2}). Let us denote the set of vertices of $P^*$ by $V(P^*)$. Formally we show that
\begin{align*}
V(P^*) = \{x \in R^d \mid \exists (j_1, \dots , j_{d/2}):  a_i^{j_i} (x_{2i-1}, x_{2i})^T = b_i^{j_i} \wedge a_i^{j_i + 1} (x_{2i-1}, x_{2i})^T = b_i^{j_i+1} \forall i\} .
\end{align*}
\begin{figure*}[h]
\begin{center}
\includegraphics{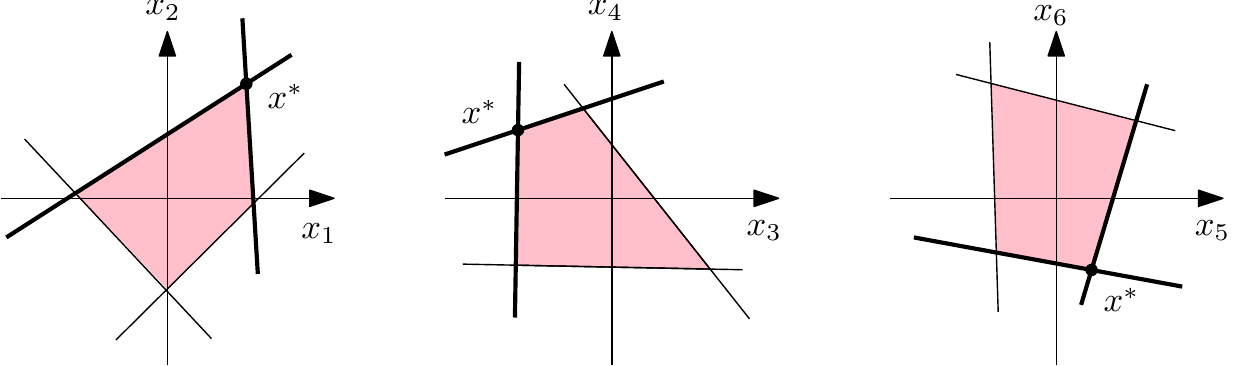}
\end{center}
\caption{$d$ constraints that define vertex in $P^*(12, 6)$} \label{fig_ex2}
\end{figure*}
Let us first show that the set on the right hand side is a subset of $V(P^*)$. We show that the $g_i^1, g_i^2$, $1 \leq i \leq d/2$, define a vertex, the rest follows from symmetry. Let us denote those $d$ constraints by $G'$ and $x^*$ the intersection point of their boundaries. It follows that $x^* \in P^*$ because $(x^*_{2i-1}, x^*_{2i}) \in P_i^*$ for all $i$.
We define the halfspace $h$ by
\[h: \sum_{i=1}^{d/2} (a_i^1 (x_{2i-1}, x_{2i})^T + a_i^2 (x_{2i-1}, x_{2i})^T) \leq \sum_{i=1}^{d/2}(b_i^1 + b_i^2),\]
the halfspace obtained by the sum of all constraints in $G'$. Let us denote this halfspace by $h: a'x \leq b'$. Then by definition if follows that $a'x^* = b'$. Now let $y \in P^* \setminus x^*$. Since $y \in P^*$ it follows that  $a_i^1 (x_{2i-1}, x_{2i})^T \leq  b_i^1$ and $a_i^2 (x_{2i-1}, x_{2i})^T \leq  b_i^2$ for all $i$. Moreover since $y \neq x^*$ there exists some $k$ such that $a_k^1 (x_{2i-1}, x_{2k})^T <  b_k^1$ or $a_k^2 (x_{2k-1}, x_{2k})^T <  b_k^2$. It follows that $a'x < b'$, hence by definition of a $0$-face, $x^*$ is a vertex.

For the other direction we need to show that no other $d$ constraints define a vertex (see also Figure \ref{fig_ex3}). If we choose more than two constraints from some $G_i$, then the intersection of their boundaries is empty. If we choose two constraints in $G_i$ that are not adjacent, the point it defines in $P^*_i$ violates some constraints of $G_i$. Hence, we need to choose two consecutive constraints.
\begin{figure*}[h]
\begin{center}
\includegraphics{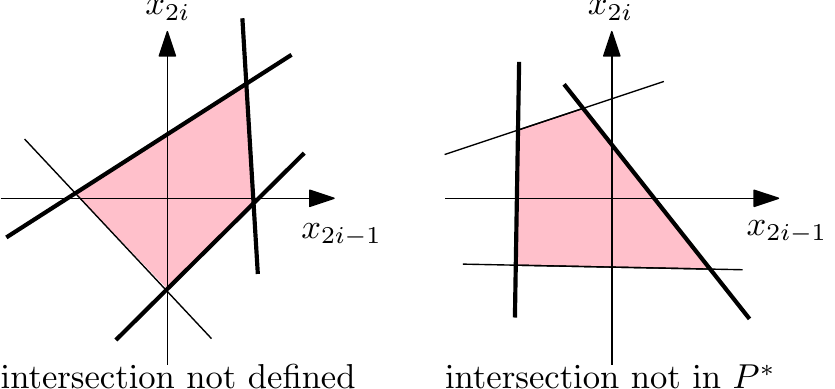}
\end{center}
\caption{$d$ constraints that do not define vertex} \label{fig_ex3}
\end{figure*}
The case where $d$ is odd is similar. The vertices of $P^*(n,d)$ are given by the constraints defining the vertices of $P^*(n-1,d-1)$ together with $x_d \geq 0$. $P^*(n-1,d-1)$ is the $d-1$ dimensional polytope defined by the constraints $G \setminus \{x_d \geq 0\}$.
The proof now follows by simple counting.
\end{proof}

We will compare the number of vertices between $P^*(n,d)$ and the dual cyclic polytope. Since we do not compare the exact values, but only the leading terms, we will not exactly compute the polynomial terms in $d$, but denote them by $\text{poly(d)}$. 
\begin{lemma}
\label{lemma_vert}
The dual cyclic polytope has a factor $O(e^{\lfloor d/2 \rfloor})$ more vertices than $P^*(n,d)$, i.e.,
$f_0(c^*(n,d)) \leq O( e^{\lfloor d/2 \rfloor}) \cdot f_0 (P^*(n,d)).$
\end{lemma}
We see that this factor is independent of $n$, hence if $d$ is constant then the number of vertices of $P^*(n,d)$ is asymptotically equal to the number of vertices of the dual cyclic polytope.  

\begin{proof}
Considering only the leading term of $f_0(c^*(n,d))$ and using inequality (\ref{eq_2}) we get
\begin{align*}
 f_0(c^*(n,d)) 
 &=  \binom{n- \lceil d/2 \rceil}{n-d} + \binom{n- \lfloor d/2 \rfloor -1}{n-d} \\
 &\leq 2 \cdot  \binom{n- \lceil d/2 \rceil}{\lfloor d/2 \rfloor} \\
 &\leq O(1) \cdot e^{\lfloor d/2 \rfloor} \cdot \left( \frac{n - \lceil d/2 \rceil}{\lfloor d/2 \rfloor}\right)^{\lfloor d/2 \rfloor} \\
 &\leq O(1) \cdot e^{\lfloor d/2 \rfloor} \cdot \left( \frac{n}{\lfloor d/2 \rfloor}\right)^{\lfloor d/2 \rfloor}.
 \end{align*}
 
 
\noindent Therefore 
 $f_0(c^*(n,d)) \leq O(e^{\lfloor d/2 \rfloor}) \cdot f_0(P^*(n,d)).$
\end{proof}

In the following we do not only compare the number of vertices between $P^*(n,d)$ and $c^*(n,d)$, but also their $f$-vectors. We will see that if $k \leq \lceil d/2 \rceil -1$, then  $f_k(P^*(n,d))$ is by a factor at most $e^{\lfloor d/2 \rfloor}$ larger than $f_k(c^*(n,d))$. If  $k \geq \lceil d/2 \rceil$, then the factor is at most $e^{d-k}$.

\begin{theorem}
\label{thm_compl}
For $d$ even 
\[f_k(P^*(n,d)) = \sum_{r = \max\{0, d/2-k\}}^{d/2 - \lfloor k/2 \rfloor}\binom{d/2}{r}\binom{d/2-r}{d -k - 2r}\left(\frac{n}{d/2}\right)^{d-k-r}.\]
For $d$ odd and $0 < k < d$
\begin{align*}
f_k(P^*(n,d)) 
&= f_{k-1}(P^*(n-1,d-1)) + f_{k}(P^*(n-1,d-1)) \\
&= \sum_{r = \max\{0, \lfloor d/2 \rfloor-(k-1)\}}^{\lfloor d/2 \rfloor - \lfloor (k-1)/2 \rfloor}\binom{\lfloor d/2 \rfloor}{r}\binom{\lfloor d/2 \rfloor-r}{d -k-2r}\left(\frac{n-1}{\lfloor d/2 \rfloor}\right)^{d-k-r}\\
&\hspace{1cm}  + \sum_{r = \max\{0, \lfloor d/2 \rfloor-k\}}^{\lfloor d/2 \rfloor - \lfloor k/2 \rfloor}\binom{\lfloor d/2 \rfloor}{r}\binom{\lfloor d/2 \rfloor-r}{(d-1) - k - 2r}\left(\frac{n-1}{\lfloor d/2 \rfloor}\right)^{(d-1)-k-r}.
\end{align*}
The value of $f_0(P^*(n,d))$ follows from Theorem \ref{thm_vert} and obviously $f_d(P^*(n,d)) = 1$. 
\end{theorem}

\begin{proof}
The proof is similar to the proof of Theorem \ref{thm_vert}, we only give the main idea. 
Assume that $d$ is even and $0 \leq k \leq d$. The $k$-faces of $P^*(n,d)$ are induced by certain intersections of $d-k$ constraints of $G$ with $P^*(n,d)$. Let $K$ be $d-k$ constraints from $G$ such that the following holds. For every $1 \leq i \leq d/2$, $G_i$ contains at most two constraints of $K$. If it contains two constraints $h_\ell$ and $h_m$ then they are consecutive, i.e., they define a vertex in $P^*_i$ (see also Figure \ref{fig_ex4}). The intersection of the boundaries of the constraints $K$ with $P^*(n,d)$ are in one to one correspondence with the $k$-faces. This works with a similar argument as in the proof of Theorem \ref{thm_vert}. 
\begin{figure*}[h]
\begin{center}
\includegraphics{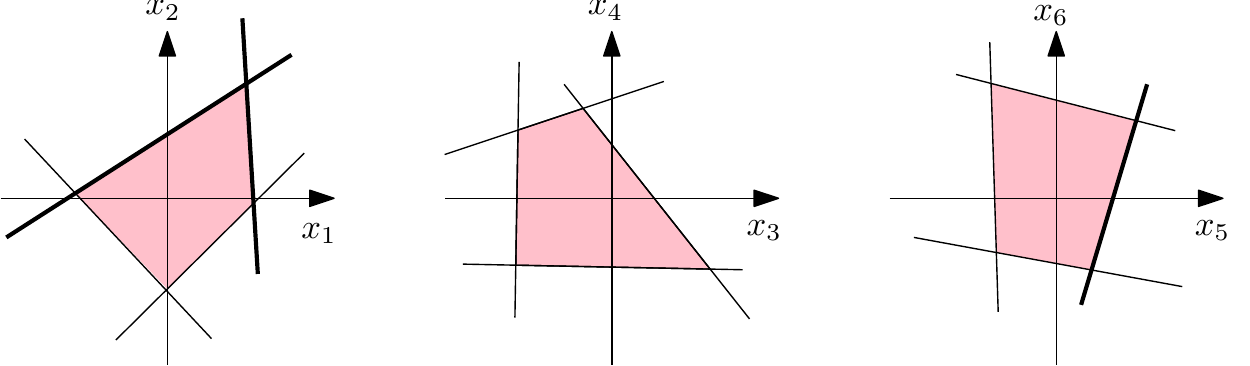}
\end{center}
\caption{Example of $3$-face in $P^*(12,6)$} \label{fig_ex4}
\end{figure*}



It remains to count the number of faces that are induced by constraints of form $K$. 
Let us consider the sets in $K$, such that there are exactly $r$ many $G_i$'s that contain two constraints of $K$. There are 
\[\binom{d/2}{r} \binom{d/2-r}{(d-k)-2r}\left(\frac{n}{d/2}\right)^{(d-k)-r}\]
of those. Now if $(d-k) \leq d/2$, then $r$ can be in $\{0, \dots, \lfloor (d-k)/2 \rfloor \}$ and if $(d-k) > d/2$, then $r$ is in $\{d/2 - k, \dots,  \lfloor (d-k)/2 \rfloor$. The claim for even $d$ follows.

The case where $d$ is odd is similar. We will not go into detail but only give the main idea. With the same kind of argumentation as above one can show that the $k$-dimensional faces are induced by $(d-k)$ constraints $K$ of $P^*$ in one of the following ways. In the first case $K$ does not contain the constraint $x_d \geq 0$. Then the constraints in $K$ must induce a $k-1$ face in $\R^{d-1}$, which then induces a $k$-face in $\R^d$. There are $f_{k-1}(P^*(n-1,d-1))$ constraints of this form. In the second case $K$ contains the constraint $x_d \geq 0$. Then the remaining $(d-k)-1$ constraints must induce a $k$ face in $R^{d-1}$. There are $f_{k}(P^*(n-1, d-1))$ constraints of this form. 
\end{proof}

\begin{lemma}
The following tables show the leading terms of $P^*(n,d)$ and $c^*(n,d)$, if $d = o(n)$. 
\end{lemma}

\begin{figure*}[h]   
\begin{center}
\begin{tabular}{ l|l|l| }
   & $P^*(n,d)$ & $c^*(n,d)$\\ \hline
  $k \leq d/2$ & $\binom{d/2}{k}\cdot \left(\frac{n}{d/2}\right)^{d/2}$ & $\binom{d/2}{k} \cdot \binom{n-d/2-1}{d/2}$\\ \hline
  $k > d/2$ & $\binom{d/2}{d-k}\cdot \left(\frac{n}{d/2}\right)^{d-k}$ & $\binom{n-k-1}{d-k}$\\
  \hline
    \multicolumn{3}{c}{$d$ even} \\
\end{tabular}
\end{center}
\label{fig_even}
\end{figure*}

\begin{figure*}[h]   
\begin{center}
\begin{tabular}{ l|l|l| }
   & $P^*(n,d)$ & $c^*(n,d)$\\ \hline
  $k \leq \lfloor d/2 \rfloor$ & $\left(\binom{\lfloor d/2 \rfloor}{k} + \binom{\lfloor d/2 \rfloor}{k+1}\right)\cdot \left(\frac{n-1}{\lfloor d/2 \rfloor}\right)^{\lfloor d/2 \rfloor}$ & $\binom{\lceil d/2 \rceil}{k} \cdot \binom{n-\lceil d/2 \rceil-1}{\lfloor d/2 \rfloor}$\\ \hline
  $k \geq \lceil d/2 \rceil $ & $\binom{\lfloor d/2 \rfloor}{d-k}\cdot \left(\frac{n-1}{\lfloor d/2 \rfloor}\right)^{d-k}$ & $\binom{n-k-1}{d-k}$\\
  \hline
    \multicolumn{3}{c}{$d$ odd} \\
\end{tabular}
\end{center}
\label{fig_odd}
\end{figure*}

The proof of the lemma follows by checking the formulas of $P^*(n,d)$ and $c^*(n,d)$.
\begin{theorem}
\label{thm_mainLI2}
For $k < \lceil d/2 \rceil$
\[f_k(c^*(n,d)) =O(\text{poly}(d) \cdot e^{\lfloor d/2 \rfloor})\cdot f_k(P^*(n,d)),\]
and for $k \geq \lceil d/2 \rceil$
\[f_k(c^*(n,d)) = O(\text{poly}(d) \cdot e^{d-k})\cdot f_k(P^*(n,d)),\]
where $\text{poly}(d)$ is some polynomial in $d$.
\end{theorem}


\begin{proof}
We only consider the case where $k \leq \lceil d/2 \rceil -1$, as the case where $k \geq  \lceil d/2 \rceil$ is similar. We will prove the statement for odd $d$, the case where $d$ is even follows immediately by replacing all $\lfloor d/2 \rfloor$ and $\lceil d/2 \rceil$ by $d/2$.
First note that the leading term of $P^*(n,d)$ can be written as
\begin{align*}
\left(\binom{\lfloor d/2 \rfloor}{k} + \binom{\lfloor d/2 \rfloor}{k+1}\right)\cdot \left(\frac{n-1}{\lfloor d/2 \rfloor}\right)^{\lfloor d/2 \rfloor} = \text{poly}(d)\cdot\binom{\lceil d/2 \rceil}{k}\cdot \left(\frac{n-1}{\lfloor d/2 \rfloor}\right)^{\lfloor d/2 \rfloor}.
\end{align*}
For the term of $c^*(n,d)$ we have
\begin{align*}
\binom{\lceil d/2 \rceil}{k} \cdot \binom{n-\lceil d/2 \rceil-1}{\lfloor d/2 \rfloor} &\stackrel{(\ref{eq_2})}{\leq}  O(1) \cdot e^{\lfloor d/2\rfloor}\cdot \binom{\lceil d/2 \rceil}{k}\cdot\left(\frac{n-\lceil d/2 \rceil-1}{\lfloor d/2 \rfloor}\right)^{\lfloor d/2\rfloor}  \\
 & \leq O(1) \cdot e^{\lfloor d/2\rfloor}\cdot\binom{\lceil d/2 \rceil}{k}\cdot \left(\frac{n-1}{\lfloor d/2 \rfloor}\right)^{\lfloor d/2\rfloor}.
\end{align*}
Therefore if $k \leq \lceil d/2 \rceil -1$,
it holds that $f_k(c^*(n,d)) = O(\text{poly}(d)\cdot e^{\lfloor d/2 \rfloor}) \cdot f_k(P^*(n,d)).$
\end{proof}

\section{Upper Bound on Maximum Complexity of $LI(2)$}
In this section we show that no polytope in $LI(2)$ can achieve the complexity of the dual cyclic polytope. To our knowledge, this is the first time such bounds are given. In Lemma \ref{lemma_half} we show that for all polytopes $P$ in $LI(2)$, $d \geq 4$ and $\lceil d/2 \rceil \leq k \leq d-2$ it holds that
$f_k(P) < f_k(c^*(n,d))$. Using this result in Theorem \ref{thm_slack}, we show that this holds for all $k \leq d-2$.

\begin{lemma}
\label{lemma_half}
Let $P$ be any polytope in $LI(2)$ given by $n$ nonredundant constraints, $d \geq 4$ and denote by $n'$ the number of constraints that contain exactly two variables per inequality. As for each index $i \in [d]$ there are at most two inequalities that contain only $x_i$ it follows that $n- 2d \leq n' \leq n$. Then for all $\lceil d/2 \rceil \leq k \leq d-2$ we have
\[f_k(P) < f_k(c^*(n,d)).\]
In particular
\[f_{d-2}(P) \leq \binom{n}{2} - \frac{\binom{n'}{2}}{\binom{d}{2}} + n' < \binom{n}{2} = f_{d-2}(c^*(n,d).)\]
\end{lemma}

\begin{proof}
Let us focus on the case of $f_{d-2}(P)$. In the dual cyclic polytope we know that any two facets are adjacent, i.e., their intersection defines a $(d-2)$-face. In $LI(2)$ however, not every two facets can be adjacent. Assume $P$ is given by $n$ constraints with index set $E$. For $i<j \in [d]$ let $E_{ij}$ be the indices of the constraints that contain $x_i$ and $x_j$ and denote $|E_{ij}| = n_{ij}$. As in the proof of Theorem \ref{thm_compl} we know that out of the $\binom{n_{ij}}{2}$ pairs only $n_{ij}$ pairs are adjacent. Summing over all $i<j$ it follows that at least
\[\sum_{i<j} \left(\binom{n_{ij}}{2}- n_{ij}\right)\]
pairs of facets in $P$ are not adjacent. 
Now using that $\sum_{i<j} n_{ij} = n'$ and that the sum is minimized if all $n_{ij}$ have the same size $n'/ \binom{d}{2}$, we get that 
\begin{align*}
\sum_{i<j} \left(\binom{n_{ij}}{2}- n_{ij}\right)  
&\geq \binom{d}{2}\cdot \binom{\frac{n'}{\binom{d}{2}}}{2} - n' \\
& = \frac{1}{\binom{d}{2}} \cdot \frac{n'\cdot\left(n'- \frac{1}{\binom{d}{2}}\right)}{2} - n' \\
&\geq \frac{\binom{n'}{2}}{\binom{d}{2}} -n'.
\end{align*}
The claim for $k=d-2$ follows. For other values of $k$ one can similarly show that not all $(d-k)$-tuples of constraints define a $k$-face in $P$. 
\end{proof}

\begin{theorem}
\label{thm_slack}
Let $P$ be any $d$-dimensional polytope in $LI(2)$ given by $n$ nonredundant constraints, where $d \geq 4$. Then for all $ k \leq d-2$ we have
\[f_k(P) < f_k(c^*(n,d)).\]
In particular 
\[f_{k}(P) \leq f_k(c^*(n,d)) - \binom{d-2}{k}\cdot\left(\frac{\binom{n'}{2}}{\binom{d}{2}} + n'\right),\]
where $n'$ is defined as in Lemma \ref{lemma_half}.
\end{theorem}

Although asymptotically the bounds that we prove are the same as the bounds of the dual cyclic polytope, this shows that the dual cyclic polytope is not realizable in $LI(2)$. 

Before proving this theorem we introduce a few notions used in the proof of McMullen's Upper Bound Theorem (for more details see  \cite{McMullen, Kalai, Fukuda2016}). From now on we only consider simple $d$-dimensional polytopes given by $n$ nonredundant constraints. A polytope $P$ is called \emph{simple}, if every vertex of $P$ is satisfies exactly $d$ inequalities with equality. We observe that by small perturbations, for any $d$-dimensional $P'$ in $LI(2)$ given by $n$ inequalities there exists a simple polytope $P$ in $LI(2)$ with $f_k(P') \leq f_k(P)$ for all $k \in [d]$. Let us denote the family of simple $d$-dimensional polytopes in $LI(2)$ by $SLI(2)$.

Let $P$ be any polytope in $SLI(2)$, given by $n$ nonredundant constraints. We consider a linear program with objective value $c^Tx$, subject to those constraints. We assume that $c$ is generic, i.e., no edge of $P$ is parallel to the hyperplane given by $c^Tx = 0$. We now orient every edge of $P$ w.r.t.\ $c^Tx$, towards the vertex with higher objective value. Let us denote the graph defined by those directed edges by $\overrightarrow{G}(P)$. Now for $i = 0, \dots , d$ we denote by $h_i(\overrightarrow{G}(P))$ the number of vertices with indegree $i$.

By double counting one can show that $h_i(\overrightarrow{G}(P))$ is independent of the objective value, hence we can write $h_i(\overrightarrow{G}(P)) = h_i(P)$ 
Let $k$ be fixed, we count the pairs $(F, v)$ of $k$ faces $F$ with unique sink $v$. By definition of $\overrightarrow{G}(P)$ every face has a unique sink, hence there are exactly $f_k(P)$ many such pairs. 
On the other hand by properties of simple polytopes it holds that for any $k$ distinct edges to $v$, there exists a unique $k$-face containing the $k$ edges. 
Let $v$ be fixed and let $r$ be the indegree of $v$. 
Summing over all indegrees $r \geq k$ it follows that for all $k = 0, \dots , d$,
\begin{equation}
\label{eq_LI1}
\sum_{r=k}^d h_r(\overrightarrow{G}(P)) \binom{r}{k} = f_k(P).
\end{equation}

\noindent Solving this system of linear equalities one can show that for all $i = 0, \dots, d$,
\begin{equation}
\label{eq_LI2}
h_i(P): = h_i(\overrightarrow{G}(P)) = \sum_{k=i}^d(-1)^{k-i}\binom{k}{i}f_k(P).
\end{equation}
Hence $h_i(P)$ is independent of the objective value.

To prove Theorem \ref{thm_slack} we use the following strengthened version of McMullen's theorem, which holds for any simple polytopes. This strengthening was first given by Kalai in \cite{Kalai} with a small correction made by Fukuda in \cite[Chapter 7]{Fukuda2016}. Note that Theorem \ref{thm_strongerMM} implies McMullen's theorem, 
 since by (\ref{eq_LI1}) we know that each $f_k(P)$ is a nonnegative linear combination of the $h_r(P)$'s.
\begin{theorem} [Strengthened Upper Bound Theorem \cite{McMullen, Fukuda2016}] 
\label{thm_strongerMM}
Let $P$ be a simple polytope given by $n$ nonredundant constraints. Then for all $i = 0, \dots , d$ it holds that
\[h_i(P) \leq h_i(c^*(n,d)).\]
\end{theorem}

\begin{proof}[Proof of Theorem \ref{thm_slack}]
Let $P$ be any polytope in $SLI(2)$.  By Lemma \ref{lemma_half} the theorem holds for $\lceil d/2 \rceil \leq k \leq d-2$ (since $d \geq 4$ it holds in particular for $k= d-2$). 
We claim that
 \[h_{d-2}(P) \leq h_{d-2}(c^*(n,d))- \frac{\binom{n'}{2}}{\binom{d}{2}} + n'.\] 
By equation (\ref{eq_LI2})
\[h_{d-2}(P) = f_{d-2}(P) - (d-1)f_{d-1}(P) + \binom{d}{d-2}f_d(P).\]
We know that 
\[f_{d-1}(P) = f_{d-1}(c^*(n,d)) = n \text{ and } f_d(P) = f_d(c^*(n,d)) = 1.\]
Furthermore by Lemma \ref{lemma_half} we know 
\[f_{d-2}(P) \leq f_{d-2}(c^*(n,d)) - \frac{\binom{n'}{2}}{\binom{d}{2}} + n'.\]
It follows that 
\begin{align*}
h_{d-2}(P) &= f_{d-2}(P) - (d-1)f_{d-1}(P) + \binom{d}{d-2}f_d(P) \\
&\leq f_{d-2}(c^*(n,d))  - \frac{\binom{n'}{2}}{\binom{d}{2}} + n' - (d-1)f_{d-1}(c^*(n,d)) + \binom{d}{d-2}f_d(c^*(n,d)) \\
&= h_{d-2}(c^*(n,d)) - \frac{\binom{n'}{2}}{\binom{d}{2}} + n',
\end{align*}
which shows the claim.
By equation (\ref{eq_LI1}) and Theorem \ref{thm_strongerMM} for $k \leq d-2$ it follows that 
\begin{align*}
 f_k(P) &= \sum_{r=k}^d  \binom{r}{k}  h_r(P) \\
 &\leq   \sum_{r=k}^d \binom{r}{k} h_r(c^*(n,d)) - \binom{d-2}{k}\left(\frac{\binom{n'}{2}}{\binom{d}{2}} + n'\right) \\
 &= f_k(c^*(n,d)) - \binom{d-2}{k}\left(\frac{\binom{n'}{2}}{\binom{d}{2}} + n'\right).
\end{align*}
\end{proof}

\begin{remark}
One can show that for $d=3$, the bounds of McMullen's Upper Bound Theorem can be achieved in $LI(2)$. Let $P$ be a polytope given by $n-2$ constraints in variables $x_1$ and $x_2$, such that they define a polygon with $n-2$ vertices in two dimensions. We furthermore add the constraints $x_3 \geq 0$ and $x_3 \leq 1$. We can easily observe that $f_0 = 2n-4$ and $f_1 = 3n -6$. Those are exactly the bounds achieved by the dual cyclic polytope. 
\end{remark}

\section{Discussion and Open Questions}
We saw that $f_k(P^*(n,d))$ differs from $f_k(c^*(n,d))$ by a factor $O(e^{\lfloor d/2 \rfloor})$ if $k < \lceil d/2 \rceil$ and $O(e^{d-k})$ otherwise.
In particular, if $d$ is constant, then $f_k(P^*(n,d))$ is of the same order as $f_k(c^*(n,d))$. The high complexity of $P^*(n,d)$ shows us that although $LI(2)$ has a much simpler structure than general linear programs, it is still a powerful and complex tool. We also showed that the dual cyclic polytope is not realizable in $LI(2)$. However in the upper bound we showed, the asymptotic complexity remains the same. It would be interesting to get a deeper understanding of $LI(2)$ and how it is different from general linear programs. 
The main open question that remains is how large the complexity of $f(P)$ can be for a polytope $P$ in $LI(2)$. Is it possible to have higher complexity than the complexity of $P^*(n,d)$? If yes, what is the maximum complexity that can be achieved? Is it asymptotically the same as the complexity of the a dual cyclic polytope? This is an interesting direction for future research. 
\bibliographystyle{abbrv}
\bibliography{redundancyFVC}
\end{document}